\documentclass[11pt]{article}
\usepackage[hmargin=1in,vmargin=1in]{geometry}
\usepackage{amsmath}
\usepackage{amsfonts}
\usepackage{amsthm}
\usepackage{amssymb}
\usepackage{multirow}
\renewcommand{\=}{\approx}
\newcommand{\Va}{\mathbf{a}}
\newcommand{\Vb}{\mathbf{b}}

\newcommand{\floor}[1]{\lfloor #1 \rfloor}

\newcommand{\F}{\mathbb{F}}
\newcommand{\Fp}{\mathbb{F}_p}
\newcommand{\Fa}{\F(\eta)}
\newcommand{\N}{\mathcal{N}}
\newcommand{\C}{\mathcal{C}}

\newcommand{\R}{\mathcal{R}}
\newcommand{\Ra}{\mathcal{R}(\eta)}
\newcommand{\ket}[1]{|#1\rangle}

\newcommand{\ring}{\F[x]/(x^n-1)}

\newcommand{\ideal}[1]{\langle #1 \rangle}
\newcommand{\weight}[1]{\ensuremath{\mathrm{w}\left(#1\right)}}
\newcommand{\symp}[2]{\ensuremath{\left\langle #1,#2 \right\rangle}}
\newcommand{\transpose}[1]{\ensuremath{#1^{\mathrm{T}}}}

\newcommand{\centraliser}[1]{\ensuremath{\overline{#1}}}

\newtheorem{theorem}{Theorem}[section]
\newtheorem{lemma}[theorem]{Lemma} 
\newtheorem{definition}[theorem]{Definition}
\newtheorem{corollary}[theorem]{Corollary}
\newtheorem{proposition}[theorem]{Proposition}

\begin{document}

\title{Quantum Cyclic Code} \author{Sagarmoy Dutta\\
  Dept of Computer Science and Engineering,\\
  Indian Institute of Technology Kanpur,\\
  Kanpur, UP, India, 208016\\
  {\tt sagarmoy@cse.iitk.ac.in}
  \and Piyush P Kurur\\
  Dept of Computer Science and Engineering,\\
  Indian Institute of Technology Kanpur,\\
  Kanpur, UP, India, 208016\\ 
  and\\
  Max-Planck Institut f\"ur Informatik\\
  Campus E1 4, 66123, Saarbr\"ucken, Germany \\
  {\tt ppk@cse.iitk.ac.in, pkurur@mpi-inf.mpg.de}
}
\date{} \maketitle

\newcommand{\etal}[0]{\emph{et al}}

\begin{abstract}

  In this paper, we define and study \emph{quantum cyclic codes}, a
  generalisation of cyclic codes to the quantum setting.  Previously
  studied examples of quantum cyclic codes were all quantum codes
  obtained from classical cyclic codes via the CSS
  construction. However, the codes that we study are much more
  general. In particular, we construct cyclic stabiliser codes with
  parameters $[[5,1,3]]$, $[[17,1,7]]$ and $[[17,9,3]]$, all of which
  are \emph{not} CSS. The $[[5,1,3]]$ code is the well known Laflamme
  code and to the best of our knowledge the other two are new
  examples. Our definition of cyclicity applies to non-stabiliser
  codes as well; in fact we show that the $((5,6,2))$ nonstabiliser
  first constructed by Rains\etal~\cite{rains97nonadditive} and latter
  by Arvind \etal~\cite{arvind:2004:nonstabilizer} is cyclic.

  We also study stabiliser codes of length $4^m +1$ over
  $\mathbb{F}_2$ for which we define a notation of BCH distance.  Much
  like the Berlekamp decoding algorithm for classical BCH codes, we
  give efficient quantum algorithms to correct up to
  $\floor{\frac{d-1}{2}}$ errors when the BCH distance is $d$.

\end{abstract}

\section{Introduction}


One of the biggest challenge in implementation of quantum computation
is to deal with quantum errors efficiently. The subtle nature of
quantum phenomenon like entanglement and superposition needs to be
preserved from both the environment as well as from faulty circuits,
for any of the speedups to be realised. Quantum error correcting codes
provide a way to make this this possible. Despite strange phenomenons
like the no-cloning theorem, a sufficiently detailed theory of quantum
error correcting
exists\cite{knill2000theory,gottesman:1996:error,calderbank98quantum}.
It has already provided a foundation for fault-tolerant quantum
computing via the implementation of error resistant quantum circuits
and quantum storage elements. Besides it plays an important role in
various areas like quantum cryptography and quantum key distribution
protocols.

In this paper, we study a certain class of quantum codes which we
believe is a natural generalisation of the class of cyclic code in the
classical setting.  We give complete characterisation such codes and
study the stabiliser case in depth
(Section~\ref{sec:quantum-cyclic-codes}). Previously, Calderbank
\etal\cite[Section 5]{calderbank98quantum} had constructed quantum
stabiliser codes whose underlying totally isotropic set (see
Section~\ref{sec:prelims} for a definition) is cyclic. In the
literature, there has been work \cite{zhi2006qbch,Aly2006primQBCH} in
trying to use classical cyclic codes to build efficient quantum codes
via the CSS construction\cite{calderbank96css,steane96css}. As
concrete examples we construct a $[[5,1,3]]$, a $[[17,1,7]]$ and a
$[[17,9,3]]$ quantum cyclic code none of which are CSS. The 5 qubit
code is the well known Laflamme code and to the best of our knowledge
the other are new. Besides we give some examples of nonstabiliser
cyclic codes as well (Section~\ref{sec:nonstabiliser}).

We also study a restricted family of cyclic stabiliser codes for which
we can define a notion of BCH distance. Much like the classical case
these family of code have efficient (polynomial time) quantum decoding
algorithm within the BCH limit, i.e. if the BCH distance is $d=2t+1$
then we can correct up to $t$ errors.

\section{Preliminaries}\label{sec:prelims}

We now give a brief overview of the notation used in this paper. For a
prime power $q = p^k$, $\mathbb{F}_q$ denotes the unique finite field
of cardinality $q$. In this paper, we study quantum codes over the
\emph{alphabet} $\mathbb{F}_p$. Most of what we say carry over any
extension $\mathbb{F}_{p^k}$ as well.

We consider the $p$-dimensional Hilbert space
$\mathcal{H}=\mathcal{L}^2(\mathbb{F}_p)$ of all functions from
$\mathbb{F}_p$ to the set of complex numbers $\mathbb{C}$. This
Hilbert space plays the role of the alphabet set in the quantum
setting.  The set $\{ \ket{a} | a \in \mathbb{F}_p \}$ where $\ket{a}$
stands for the function that takes value $1$ at $a$ and 0 every where
else, forms an orthonormal basis for the Hilbert space
$\mathcal{H}$. For a positive integer $n$, an element $\mathbf{a} =
(a_1,\ldots,a_n)^\mathrm{T} \in \mathbb{F}_p^n$ will be considered as
column vectors. As is standard in quantum computing, by
$\ket{\mathbf{a}}$ we mean the vector
$\ket{a_1}\otimes\ldots\otimes\ket{a_n}$. Thus, the set $\{
\ket{\mathbf{a}} | \mathbf{a} \in \mathbb{F}_p^n \}$ forms a basis for
the $n$-fold tensor product $\mathcal{H}^{\otimes^n}$.

Quantum errors are captured by what are know as the \emph{Weyl
  operators}. For $a$ and $b$ in $\mathbb{F}_p$ define the unitary
operators $U_a$ and $V_b$ as $U_a \ket{x} = \ket{x+a}$ and $V_b
\ket{x} = \zeta^{bx} \ket{x}$, where $\zeta$ is a primitive $p$-th
root of unity.  The operator $U_a$ is thought of as a flip in the
alphabet $\mathbb{F}_p$ and $V_b$ is thought of as a flip in the
phase. The operator $U_aV_b$ constitutes a flip in both the alphabet
and phase. It is sufficient to consider only the Weyl operators when
designing quantum codes as they form a basis of the Hilbert space
$\mathcal{B}(\mathcal{H})$ of operators from $\mathcal{H}$ on to
itself. To extend these operators onto $\mathcal{H}^{\otimes^n}$, for
a positive integer $n$, define for $\mathbf{a}$ and $\mathbf{b}$ in
$\mathbb{F}_p^n$, the Weyl operators $U_{\mathbf{a}}$ and
$V_{\mathbf{b}}$ on $\mathcal{H}^{\otimes^n}$ as $U_\mathbf{a}
\ket{\mathbf{x}} = \ket{\mathbf{x} + \mathbf{a}}$ and $V_\mathbf{b}
\ket{\mathbf{x}} = \zeta^{\mathbf{b}^\mathrm{T} \mathbf{x}} \ket{x}$
respectively.  To capture errors at $t$ locations, we define the
\emph{joint weight} $\weight{\mathbf{a},\mathbf{b}}$ for a pair
$\mathbf{a}=(a_1,\ldots,a_n)$ and $\mathbf{b}=(b_1,\ldots,b_n)$ in
$\mathbb{F}_p^n$ as the number of positions $i$ such that either $a_i$
or $b_i$ is not zero. We extend this definition to Weyl operators, the
\emph{weight} $\weight{U_\mathbf{a}V_\mathbf{b}}$ is the joint weight
$\weight{\mathbf{a},\mathbf{b}}$. Consider the transmission of any
pure state $\ket{\psi}$ in $\mathcal{H}^{\otimes^n}$. Occurrence of a
quantum error at $t$ positions is modelled as the channel applying an
unknown Weyl operator $U_\mathbf{a}V_\mathbf{b}$ of weight $t$ on the
transmitted message $\ket{\psi}$.  An \emph{quantum code} over
$\mathbb{F}_p$ of \emph{length} $n$ is a subspace of the $n$-fold
tensor produce $\mathcal{H}^{\otimes^n}$. There is by now a
significant literature on quantum
codes\cite{knill2000theory,gottesman:1996:error,calderbank98quantum}. A
quantum code, being a subspace, is completely captured by the
projection into it. Therefore we often express a quantum code by
giving its projection operator. A quantum code of length $n$,
dimension $K$ and distance $d$ over $\mathcal{L}^2(\mathbb{F}_p)$ will
be called an $((n,K,d))_p$ quantum code.

We now discuss special quantum codes called \emph{stabiliser codes} or
\emph{additive codes}. To this end fix a finite field $\mathbb{F}_p$
and a positive integer $n$. Let $\mathcal{W}_{n,p}$ denote the group
of unitary operators generated by the Weyl operators. The \emph{error
  group} $\mathcal{E}_{n,p}$ is just the group $\mathcal{W}_{n,p}$, if
the characteristic $p$ is odd, and is the group generated by
$\mathcal{W}_{n,p} \cup \iota \mathcal{W}_{n,p}$, $\iota$ being the
complex number $\sqrt{-1}$, when $p$ is 2. We will drop the subscripts
$n$ and $p$ when the quantities are clear from the context. For a
subset $\mathcal{S}$ of the error group $\mathcal{E}$, the
\emph{stabiliser code} $\mathcal{C}_\mathcal{S}$ associated with
$\mathcal{S}$ is the subspace of vectors $\ket{\varphi} \in
\mathcal{H}^{\otimes^n}$ such that $U \ket{\varphi} = \ket{\varphi}$
for all $U$ in $\mathcal{S}$. Without loss of generality we can assume
that $\mathcal{S}$ is actually a subgroup of $\mathcal{E}$.
Furthermore, for the code $\mathcal{C}_\mathcal{S}$ to be be
nontrivial, $\mathcal{S}$ should be Abelian and should not contain
$\omega I$ for any nontrivial root of unity
$\omega$\cite{calderbank97orthogonal,arvind2003family}. We call such a
subgroup $\mathcal{S}$ a \emph{Gottesman subgroups} of the error group
$\mathcal{E}$.

Consider a pair $(\mathbf{a},\mathbf{b})$ of elements of
$\mathbb{F}_p^n$ as elements of the vector space $\mathbb{F}_p^{2n}$
over the scalars $\mathbb{F}_p$. The \emph{symplectic} inner product
between two elements $\mathbf{u} = (\mathbf{a},\mathbf{b})$ and
$\mathbf{v} = (\mathbf{c},\mathbf{d})$ of $\mathbb{F}_p^{2n}$ is
defined as $\symp{\mathbf{u}}{\mathbf{v}} =
\transpose{\mathbf{a}}\mathbf{d} - \transpose{\mathbf{b}}\mathbf{c}$.
A subset $S$ of $\mathbb{F}_p^{2n}$ is called \emph{totally
  isotropic}~\cite{calderbank98quantum} with respect to the symplectic
inner product, if for any two elements $\mathbf{u}$ and $\mathbf{v}$
of $S$, $\symp{\mathbf{u}}{\mathbf{v}} = 0$. In the rest of the paper,
unless otherwise mentioned, by totally isotropic set we mean totally
isotropic with respect to the symplectic inner product define above.

Stabiliser codes, or equivalently their corresponding Gottesmann
subgroups, are intimately connected to totally isotropic subspace of
$\mathbb{F}_p^{2n}$. Depending on whether the characteristic $p$ is
odd or $2$, the exact nature of the correspondence is slightly
different. Calderbank \etal~\cite{calderbank97orthogonal,calderbank98quantum} were the
first to study this connection when characteristic of the underlying
field $\mathbb{F}_p$ is $2$.  Later Arvind and
Parthasarathy\cite{arvind2003family} studied the case when $p$ is an
odd prime.  We summaries these results in a form convenient for our
purposes.

\begin{theorem}[\cite{calderbank97orthogonal,arvind2003family}]
  \label{thm:iso-stab-connection}
  Let $p$ be any prime and $n$ a positive integer. If $S$ is a totally
  isotropic subspace of $\mathbb{F}_p^{2n}$ there exists $n \times 2n$ matrices $L$ and $M$ such that
  \begin{enumerate}
  \item $\transpose{L}M$ is symmetric,
  \item $S$ is the image of the map $\phi_{L,M}$ from
    $\mathbb{F}_p^n$ to $\mathbb{F}_p^{2n}$ defined as
    $\phi_{L,M}(\mathbf{a}) = (L\mathbf{a},M\mathbf{a})$ and
  \item The set of operators $\mathcal{S} = \{ \alpha_{\mathbf{a}}
    U_{L\mathbf{a}} V_{M\mathbf{a}} | \mathbf{a} \in \mathbb{F}_p^n
    \}$ forms a Gottesman subgroup where $\alpha_\mathbf{a}$ is defined as
    \[
    \alpha_\mathbf{a} = \left\{ \begin{array}{ll}
        \zeta^{\frac{1}{2}\transpose{\mathbf{a}}\transpose{L}M\mathbf{a}} & \textrm{ when } p \neq 2,\\
        \iota^{\transpose{\mathbf{a}}\transpose{L}M\mathbf{a}} &
        \textrm{ when } p = 2
      \end{array}\right.
    \]
    where $\zeta$ is a primitive $p$-th root of unity and $\iota$ is
    $\sqrt{-1}$.
    \item The projection operator to the associated code $\mathcal{C}_\mathcal{S}$ is
      given by 
      \[
      P = \sum_{U \in \mathcal{S}} U = \sum_{\mathbf{a}}
      \alpha_\mathbf{a} U_{L\mathbf{a}} V_{M\mathbf{a}}.
      \]
  \end{enumerate}
\end{theorem}

When studying stabiliser codes, we will concentrate only on the
underlying totally isotropic subspace of $\mathbb{F}_p^{2n}$. 

One possible way of constructing quantum codes, or totally isotropic
subspaces of $\mathbb{F}_p^n \times \mathbb{F}_p^n$, is by taking
classical codes $C_1$ and $C_2$ of length $n$ such that $C_1$ is
orthogonal to $C_2$ (here the orthogonaity is with respect to the
usual inner product $\transpose{\mathbf{a}}\mathbf{b}$). It is easy to
verify then that $C_1\times C_2$ is isotropic. This construction is
called the CSS construction~\cite{calderbank96css} and the resultant
quantum stabiliser codes are called CSS codes.

Let $S$ be a subspace of $\mathbb{F}_p^{2n}$. By the
\emph{centraliser} of $S$, denoted by $\centraliser{S}$, we mean the
subspace of all $\mathbf{u}$ in $\mathbb{F}_p^{2n}$, such that
$\symp{\mathbf{u}}{\mathbf{v}} = 0$, for all $\mathbf{v}$ in $S$. If
$S$ is totally isotropic, $\centraliser{S}$ contains $S$. We have the
following theorem on the properties of the stabiliser code
$\mathcal{C}_S$ associated to the set $S$.

\begin{theorem}[\cite{calderbank97orthogonal,arvind2003family}]
  \label{thm:isotropic-main-properties}
  Let $S$ be a totally isotropic subspace of $\mathbb{F}_p^{2n}$ and
  let $\mathcal{C}$ be the associated stabiliser code. Then, the
  dimension the subspace $S$ is always less than $n$. If $S$ has
  dimension $n - k$ for some positive integer $k$ then dimension of
  its centraliser $\centraliser{S}$, as a subspace of
  $\mathbb{F}_p^{2n}$, and the dimension of the code $\mathcal{C}$, as
  a Hilbert space, are $n+k$ and $p^{k}$ respectively. Furthermore, if
  the minimum weight $\textrm{min}\{\weight{\mathbf{u}} | \mathbf{u}
  \in \centraliser{S} \setminus S \}$ is $d$ then $\mathcal{C}$ can
  detect upto $d-1$ errors and correct up to $\floor{\frac{d - 1}{2}}$
  errors.
\end{theorem}

Let $\mathcal{C}$ be a stabiliser code associated with an $n-k$
dimensional totally isotropic subspace $S$ of $\mathbb{F}_p^{2n}$. By
the \emph{stabiliser dimension} of $\mathcal{C}$ we mean the integer
$k$. Similarly, we call the weight $\textrm{min}\{ \weight{\mathbf{u}}
| \mathbf{u} \in \centraliser{S} \setminus S \}$ the \emph{distance} of
$\mathcal{C}$.  A stabiliser code of length $n$, stabiliser dimension
$k$ and distance $d$ will be called an $[[n,k,d]]_p$ stabiliser
code. By theorem~\ref{thm:isotropic-main-properties}, an $[[n,k,d]]_p$
stabiliser code is an $((n,p^k,d))_p$ quantum code. As usual we will
drop the subscript $p$ when it is clear from the context.

\section{Quantum Cyclic codes}\label{sec:quantum-cyclic-codes}

In this section we define quantum cyclic codes and study some of its
properties. Recall that a classical code over $\mathbb{F}_p$ is cyclic
if and only if for all code words $\mathbf{u} = (u_1,\ldots,u_n)$, its
right shift $(u_n,u_1,\ldots,u_{n-1})$ is also a code word.  Let $N$
denote the right shift operator over $\mathbb{F}_p^n$, i.e. the
operator that maps $\mathbf{u} = (u_1,\ldots,u_n)$ to
$(u_n,u_1,\ldots,u_{n-1})$. Consider the unitary operator $\N$ defined
on the tensor product $\mathcal{H}^{\otimes^n}$ as follows $\N
\ket{\mathbf{u}} = \ket{N \mathbf{u}}$.

\begin{definition}\label{def_qcc} 
  A quantum code $\mathcal{C}$ is defined to be \emph{cyclic} if the
  shift operator $\N$ maps $\mathcal{C}$ to itself, i.e. $\N
  \mathcal{C} = \mathcal{C}$.
\end{definition}

We have the following result on the projection operator associated to
a quantum cyclic code.

\begin{proposition}
  A quantum code $\C$ is cyclic if and only if its projection operator commutes
  with $\N$.
\end{proposition}
\begin{proof}

  Let $\mathcal{H}$ be any Hilbert space and let $\mathcal{C}$ be any
  subspace with the associated projection operator being $P$. Let $U$
  be any unitary operator on $\mathcal{H}$.  If $UP = PU$ then $U
  \mathcal{C} = U P \mathcal{H} = P U\mathcal{H}$. However, since
  $\mathcal{H}$ is the underlying Hilbert space we have $U \mathcal{H}
  = \mathcal{H}$. Thus $U \mathcal{C} = \mathcal{C}$.  

  Conversely, suppose that $U \C = \C$. We prove $UP = PU$ by showing
  that for all $\ket{\psi}$ in $\mathcal{H}$, $UP \ket{\psi} =
  PU\ket{\psi}$. Let $\C^\perp$ be the orthogonal complement of
  $\C$. Since any unitary map preserves inner product, we have $U
  \mathcal{C}^\perp = \mathcal{C}^\perp$.

  Any vector $\ket{\psi}$ can be expressed uniquely as $\ket{\psi_1} +
  \ket{\psi_2}$ where $\ket{\psi_1} \in \C$ and $\ket{\psi_2} \in
  \C^\perp$. Therefore,
  \[
  UP \ket{\psi} = U \ket{\psi_1} = P (U\ket{\psi_1} + U \ket{\psi_2}) =
  PU\ket{\psi}.
  \] Thus if $U\mathcal{C} = \mathcal{C}$ then $UP = PU$. The result
  then follows by taking $U = \N$.
  
\end{proof}

Let $S$ be a subspace of $\mathbb{F}_p^n \times \mathbb{F}_p^n$. We
say that $S$ is \emph{separately cyclic} if for all $(\Va,\Vb)$ in $S$
$(N\Va,N\Vb)$ is also in $S$. We have the following property on
centralisers of separately cyclic sets.

\begin{proposition}\label{prop:centraliser-cyclic}
  Let $S$ be any separately cyclic set then its centraliser is also
  separately cyclic.
\end{proposition}

In the context of cyclic stabiliser codes, separately cyclic sets are
interesting because of the following property.

\begin{proposition}\label{prop:stab-cyclic} 
  A stabiliser code $\mathcal{C}$ is cyclic if and only if its
  associated totally isotropic subspace $S$ and its centraliser
  $\centraliser{S}$ are separately cyclic.
\end{proposition} 

\begin{proof} 
  Let $S$ be the totally isotropic set associated with $\mathcal{C}$.
  Let $P$ denote the projector to $\mathcal{C}$. Then $\mathcal{C}$ is
  cyclic if and only if $\N^\dag P \N = P$.

  We make use of Theorem~\ref{thm:iso-stab-connection} for the proof.
  The projector $P$ is given by the expression $P = \sum_{\mathbf{a}}
  \alpha_{\mathbf{a}} U_{L\mathbf{a}}V_{M\mathbf{a}}$ and $S$ is $\{
  (L\mathbf{a},M\mathbf{a}) | \mathbf{a} \in \mathbb{F}_p^n\}$ for
  $L$, $M$ and $\alpha$ as in Theorem~\ref{thm:iso-stab-connection}.
  Since $\N^\dag U_{\mathbf{a}}V_{\mathbf{b}} \N = U_{N\mathbf{a}}
  V_{N\mathbf{b}}$, it is necessary that $S$ is separately
  cyclic. Otherwise the support of $\N^\dag P \N$ will not match with
  that of $P$.
  
  Conversely, if $S$ is separately cyclic then we have $(NL\mathbf{a},
  NM\mathbf{a}) \in S$ for all $\mathbf{a} \in \mathbb{F}_p^n$ where
  $L$ and $M$ are as in Theorem~\ref{thm:iso-stab-connection}. Also
  note that the inverse of the shift operation $N$ is just
  $\transpose{N}$. There fore $\transpose{L}\transpose{N}NM =
  \transpose{L}M$. Hence the scalars $\alpha_\mathbf{a}$ are also
  preserved and hence $\N^\dag P \N = P$.
  
  The cyclicity of the centraliser $\centraliser{S}$ follows from
  Proposition~\ref{prop:centraliser-cyclic}.

\end{proof}

Classical cyclic codes over $\mathbb{F}_p$ of length $n$, $n$ coprime
to $p$, are ideals of the polynomial ring $\mathbb{F}_p[X]/(X^n -
1)$. The goal of the rest of the section is to develop an algebraic
characterisation of cyclic stabiliser codes along similar lines.  We
fix some conventions for the rest of the paper. Fix a prime $p$ and a
positive integer $n$ coprime to $p$. Let $\R$ denote the cyclotomic
ring $\mathbb{F}_p[X]/X^n -1$ of polynomials modulo $X^n -1$.  For the
vector $\mathbf{a}=(a_0,\ldots,a_{n-1})$ in $\mathbb{F}_p^n$,
associate the polynomials $a(X) = a_0 + \ldots + a_{n-1} X^{n-1}$ in
the ring $\R$. Often, we need to interchange between these two
perspectives of an element in $\mathbb{F}_p^n$. When we think of them
as a vector, we use the bold face Latin letter. On the other hand,
when thinking of them as polynomials we use the corresponding plain
face letter. For example the polynomial associated with the vector
$\mathbf{a}$ is either written as $a(X)$ or often just $a$. In the
ring $\R$, the polynomial $X$ has a multiplicative inverse namely
$X^{n-1}$. Often, we just write $X^{-1}$ or just $\frac{1}{X}$ to
denote this inverse.

First we have the following characterisation of separately cyclic
subspaces of $\mathbb{F}_p^n \times \mathbb{F}_p^n$ in terms of
polynomials in $\R$.

\begin{lemma}\label{lem:sepera} 
  A subspace $S$ of $\mathbb{F}_p^n \times \mathbb{F}_p^n$ is
  separately cyclic if and only if there exists degree $n-1$
  polynomials $g(X)$, $f(X)$ and $h(X)$ in $\mathbb{F}_p[X]$ such that
  $g(X)$ and $h(X)$ are factors of $X^n -1$ as polynomials in
  $\mathbb{F}_p[X]$ and elements of $S$ as a pair of polynomials in
  $\R\times \R$ are precisely $(a(X)g(X), a(X) f(X) + b(X)h(X))$ where
  $a(X)$ and $b(X)$ vary over all degree $n-1$ polynomials in
  $\mathbb{F}_p[X]$.
\end{lemma}
\begin{proof}
  Consider any subspace $S$ of $\mathbb{F}_p^n \times
  \mathbb{F}_p^n$. Clearly if there exists polynomials satisfying the
  above mentioned conditions, then $S$ is separately cyclic.

  To prove the converse, assume that $S$ is separately cyclic.  Define
  $A$ and $B$ to be the projections of $S$ onto the first and last $n$
  coordinates respectively, i.e. $A=\{ \mathbf{a} |
  (\mathbf{a},\mathbf{b}) \in S \}$ and $B = \{ \mathbf{b} |
  (\mathbf{a},\mathbf{b}) \in S \}$. Since $S$ is separately cyclic,
  $A$ and $B$ are cyclic subspaces of $\mathbb{F}_p^n$ and hence are
  ideals of the ring $\R$. Let $g(X)$ be the factor of $X^n - 1$ that
  generates $A$. Since $g(X)$ is an element of $A$ there exists a
  polynomial $f(X)$ in $\R$ such that $(\mathbf{g},\mathbf{f}) \in
  S$. Fix any such polynomial $f$. To construct $h(X)$, consider the
  set $B_0=\{b(x)|(0,\Vb)\in S\}$. Clearly $B_0$ is also cyclic and
  therefore and ideal of the ring $\R$. Let $h(X)$ denote the factor
  of $X^n -1$ that generates $B_0$. Our claim is that these are the
  required polynomials. 

  Since $S$ is separately cyclic we have $(X^i g(X), X^i f(X))$ are
  all elements of $S$. As $S$ is a subspace, by taking appropriate
  linear combinations, we have that for any two degree $n-1$
  polynomials $a(X)$ and $b(X)$, $(ag , a f + bh)$ is an element of
  $S$. On the other hand, consider any arbitrary $(u,v) \in
  S$. Clearly $\mathbf{u}$ is an element of $A$ and hence $u(X) = a(X)
  g(X)$.  Subtract from the $(u,v)$ the element $(ag, af)\in S$. We
  have $(0,v - af)$ is in $S$ and hence $v - af$ is in
  $B_0$. Therefore $v - af = bh$ for some polynomial $b$.

\end{proof}

The triple of polynomials $(g,f,h)$ play a crucial role in unravelling
the structure of a separately cyclic subspace $S$. We have the
following definition.

\begin{definition}[Generating triple]
  Let $S$ be a separately cyclic subspace of $\mathbb{F}_p^n \times
  \mathbb{F}_p^n$. A \emph{generating triple} for $S$ is a triple of
  polynomials $(g,f,h)$ in $\mathbb{F}_p[X]$ the polynomials pairs
  $(g,f)$ and $(0,h)$ are in $S$ and every element of $S$ as a
  polynomial pair in $\R \times \R$ is of the form $(ag, af + b h)$
  for some polynomials $a(X)$ and $b(X)$ in $\mathbb{F}_p[X]$.
\end{definition}

We want to express the isotropic condition
$\transpose{\mathbf{a}}\mathbf{d} = \transpose{\mathbf{b}}\mathbf{c}$
in terms of polynomials. The following definition on pairs of
polynomials over $\mathbb{F}_p$ that will play the role of the
isotropic condition in the setting of separately cyclic subspace.

\begin{definition}[Isotropic pairs of polynomial]
  Let $a(X)$, $b(X)$, $c(X)$ and $d(X)$ are polynomials in
  $\mathbb{F}_p[X]$.  We say that the pairs $(a,b)$ and $(c,d)$ are
  \emph{isotropic pairs of polynomial} modulo $X^n - 1$ for some $n$
  coprime to $p$ if and only if
  \[
  a(X) d(X^{-1}) = b(X) c(X^{-1}) \mod X^n - 1 .
  \]
\end{definition}

Notice that for any two vectors $\mathbf{u}$ and $\mathbf{v}$ in
$\mathbb{F}_p^n$, if $u(X)$ and $v(X)$ denote the corresponding
polynomials in $\R$, then the coefficient of $X^k$ in the product
$u(X)v(X^{-1})\mod X^n -1$ is the inner product
$\transpose{\mathbf{u}}N^k\mathbf{v}$, where $N$ is the right shift
operator. An immediate consequence of this observation is the
following.

\begin{proposition}
  Let $S$ be a separately cyclic subspace of $\mathbb{F}_p^n \times
  \mathbb{F}_p^n$.  An element $(\mathbf{u},\mathbf{v})$ is isotropic
  to all elements of $S$ with respect to the symplectic inner product
  if and only if the corresponding pair of polynomials $(u,v)$ is
  isotropic modulo $X^n -1$ with all polynomial pairs $(a,b)$ in $S$ .
\end{proposition}

As a corollary we have the following proposition

\begin{corollary}\label{cor:isotropic-poly}
  A separately cyclic subset $S$ of $\mathbb{F}_p^{n}\times
  \mathbb{F}_p^n$ is totally isotropic if and only if for every pair
  of elements $(\mathbf{a},\mathbf{b})$ and $(\mathbf{c},\mathbf{d})$
  of $S$, the corresponding polynomials $(a,b)$ and $(c,d)$ are
  isotropic modulo $X^n -1$.  Furthermore, the centraliser
  $\centraliser{S}$ is the pair of all polynomials $(c,d)$ that are
  isotropic to all pairs of polynomials in $S$.
\end{corollary}

We are now ready to characterise all cyclic stabiliser codes. It is
sufficient to characterise separately cyclic, totally isotropic
subsets of $\mathbb{F}_p^n \times \mathbb{F}_p^n$.

\begin{theorem}\label{thm:poly-characterisation}
  A subspace $S$ of $\mathbb{F}_p^n \times \mathbb{F}_p^n$ is totally
  isotropic and separately cyclic if and only if there exists
  polynomials $g$, $f$ and $h$ in $\mathbb{F}_p[X]$ such that
  \begin{enumerate}
  \item The polynomials $g$ and $h$ divide $X^n -1$ as polynomials
    over $\mathbb{F}_p$.
  \item $g(X^{-1}) h(X) = g(X) h(X^{-1}) = 0 \mod X^n -1$,
  \item The pair $(g,f)$ is isotropic to itself modulo $X^n - 1$, i.e.
    $g(X) f(X^{-1}) = f(X) g(X^{-1})$.
  \item $S$ is precisely the set of polynomial pairs of the form $(ag,
    af + bh)$ where $a$ and $b$ varies over polynomials over
    $\mathbb{F}_p$.
  \end{enumerate}
\end{theorem}
\begin{proof}
  Assume that $S$ is both separately cyclic and totally isotropic.
  Let $(g,f,h)$ be the triple generating $S$. Since $S$ is totally
  isotropic and since $(g,f)$ is an element of $S$, it should be
  isotropic modulo $X^n -1$ with every polynomial pair in $S$.  In
  particular it should be so with itself and the element $(0,h)$. The
  shows that $g$, $f$ and $h$ satisfies the conditions in the theorem.
  
  To prove the converse, assume that polynomials $g$, $f$ and $h$
  exists with the above mentioned properties. Then clearly $S$ is
  separately cyclic. It is straight forward to then check that any two
  elements $(a_i g , a_i f + b_i h)$, for $i=1,2$ form a isotropic
  pair of polynomials modulo $X^n -1$. Hence $S$ is isotropic.
\end{proof}

A similar proof together with corollary~\ref{cor:isotropic-poly}
gives the following characterisation of the centraliser of a
separately cyclic totally isotropic subspace.

\begin{theorem}
  Let $S$ be a totally isotropic, separately cyclic subspace of
  $\mathbb{F}_p^n \times \mathbb{F}_p^n$ generated by the triple
  $(g,f,h)$ then the centraliser $\centraliser{S}$ is the set of
  all polynomial pairs $(c,d)$ such that
  \begin{enumerate}
  \item $c(X) h(X^{-1}) = c(X^{-1}) h(X) = 0 \mod X^n -1$,
  \item $(g,f)$ and $(c,d)$ forms a isotropic pairs of polynomials,
    i.e $g(X)d(X^{-1}) = f(X) c(X^{-1}) \mod X^n -1$.
  \end{enumerate}
\end{theorem}

\section{Explicit constructions and decoding algorithm}

In this section we give an explicit constructions of cyclic stabiliser
codes over the binary alphabet $\mathbb{F}_2$. We define the notion of
a BCH distance for such codes and give an efficient algorithm to
decode such codes within the BCH limit.

Consider the unique quadratic extension
$\mathbb{F}_4=\mathbb{F}_2(\eta)$ of $\mathbb{F}_2$ obtain by
adjoining a root $\eta$ of the irreducible polynomial $X^2+X+1$ in
$\mathbb{F}_2[X]$. The conjugate root $\eta+1$ of $X^2+X+1$ in
$\mathbb{F}_4$ will we denoted by $\eta'$. For this section, we
identify the set $\mathbb{F}_2^n \times \mathbb{F}_2^n$ with the
vector space $\mathbb{F}_4^n$ over $\mathbb{F}_4$ as follows: every
pair of element $(\mathbf{a}, \mathbf{b})$ where $\mathbf{a}$ and
$\mathbf{b}$ are vectors in $\mathbb{F}_2^n$ will be identified with
vector $\mathbf{a} + \eta \mathbf{b}$ in $\mathbb{F}_4$. It is
convenient to extend this identification to polynomials as well. As
before, let $\R$ denote the cyclotomic ring $\mathbb{F}_2[X]/X^n
-1$. We identify the set $\R \times \R$ with the cyclotomic ring
$\R(\eta)=\mathbb{F}_4[X]/X^n -1$ over the field extension
$\mathbb{F}_4$ by identifying the pair $(a,b)$ of polynomials in $\R
\times \R$ with the polynomial $a(X) + \eta b(X) \in \R(\eta)$. The
codes that we give in this section will in fact be linear stabiliser
codes~\cite{calderbank98quantum} defined as follows.

\begin{definition}[Linear stabiliser codes]
  A stabiliser code is said to be \emph{linear} if the underlying
  totally isotropic subspace $S$ of $\mathbb{F}_p^n \times
  \mathbb{F}_p^n$ as a subset of $\mathbb{F}_4^n$ is a
  $\mathbb{F}_4$-linear subspace of $\mathbb{F}_4^n$.
\end{definition}

We first show the following result on separately cyclic subspaces of
that are $\mathbb{F}_4$ linear.

\begin{proposition}\label{prop:seperately-cyclic-ideal}
  Let $n$ be any positive odd integer. Let $S$ be any separately
  cyclic subspace of $\mathbb{F}_2^n \times \mathbb{F}_2^n$. Then $S$
  is $\mathbb{F}_4$ linear if and only if $S$ is an ideal of
  $\mathbb{F}_4[X]/X^n -1$.
\end{proposition}
\begin{proof}
  Clearly $S$ is separately cyclic if and only if $S$ is cyclic as in
  a classical cyclic code over $\mathbb{F}_4$. The result then follows
  from the theory of classical cyclic codes.
\end{proof}

As a corollary we have the following result

\begin{corollary}
  Let $\mathcal{C}$ be any linear stabiliser code and let $S$ be the
  underlying totally isotropic set. Then $\mathcal{C}$ is cyclic if
  and only if $S$ and its centraliser $\centraliser{S}$ are ideals of
  $\mathbb{F}_4[X]/X^n -1$.
\end{corollary}

Thus constructing linear stabiliser codes that are also cyclic
involves computing factors of $X^n -1$ over $\mathbb{F}_4$ such that
the underlying ideal as a set of $\mathbb{F}_2^n \times
\mathbb{F}_2^n$ is totally isotropic.  To this end we fix some more
notation.

Consider the Forbenius automorphism $\sigma$ on $\mathbb{F}_4$ that
maps the root $\eta$ its conjugate $\eta' = \eta +1$. We extend this
to the ring $\mathbb{F}_4[X]/X^n -1$ by defined by $\sigma\left(a(X) +
\eta b(X)\right) = a(X) + \eta'b(X)$. We first state the following
result on the irreducible factors of $X^{4^m +1} - 1$ over the finite
fields $\mathbb{F}_2$ and $\mathbb{F}_4$ respectively.

\begin{lemma}\label{lem:degree} 
  Let $f(X)$ be any irreducible factor of $X^{4^m + 1} -X$ over
  $\mathbb{F}_2$.
  \begin{enumerate}
  \item A root $\beta$ of $X^n -1$ is a root of $f(X)$ if and only if
    $\beta^{-1}$ is also a root. \label{property-beta-inv}
  \item Furthermore if $f(X) \neq X -1$, the degree of $f(X)$ is an
    even number $2t$ and it splits into two irreducible factors
    $f_1(X,\eta)$ and $f_2(X,\eta)$ of degree $t$ each such that $f_2
    = \sigma(f_1)$.\label{property-even-deg}
  \end{enumerate}
\end{lemma} 
\begin{proof} 
  Let $f(X)$ be any irreducible factor of $X^{4^m +1} -1$ over
  $\mathbb{F}_2$. Consider any root $\beta$ of $f(X)$.  

  To prove Property~\ref{property-beta-inv}, note that $\sigma^{2m}$
  is a field automorphism and $\sigma^{2m}(\beta) = \beta^{4^m} =
  \beta^{-1}$. Hence $f(\beta^{-1}) = \sigma^{2m} (f(\beta)) = 0$.

  Now consider Property~\ref{property-even-deg}.  Consider any root
  $\beta$ of $f(X)$.  Since $f(X) \neq X -1$, we have $\beta \neq 1$
  and hence $\beta \neq \beta^{-1}$. As a result, that roots of $f(X)$
  comes in pairs; for every root $\beta$ its inverse $\beta^{-1}$ is
  also a root. Hence, $f(X)$ has to be of even degree.

  Let the degree of $f$ be the even number $2t$ then the splitting
  field of $f(X)$ over $\mathbb{F}_2$ is is $\mathbb{F}_{4^t}$ and
  therefore contains $\mathbb{F}_4$. Furthermore any irreducible
  factor $f_1(X)$ of $f(X)$ over $\mathbb{F}_4$ has degree exactly the
  degree of the extension $\mathbb{F}_{4^t}/\mathbb{F}_4$ which is
  $t$. Since $\sigma$ is a field automorphism of $\mathbb{F}_4$, we
  have $f(X) = f_1(X) \sigma(f_1(X))$ over $\mathbb{F}_4$.
\end{proof}

The following theorem gives us a generic way of constructing linear
cyclic codes over $\mathbb{F}_2$ of length $4^m +1$.

\begin{theorem}\label{thm:main-construction}
  Let $n = 4^m +1$ and let $g(X)$ be any factor of $X^n - 1$ over
  $\mathbb{F}_2$.  Let $h(X,\eta)$ be any product of irreducible
  factors of $X^n -1/g$ over $\mathbb{F}_4$, such that for any
  irreducible factor $r(X,\eta)$ of $X^n -1$ over $\mathbb{F}_4$,
  $r(X,\eta)$ divides $h(X,\eta)$ if and only if $\sigma(r) =
  r(X,\eta')$ \emph{does not}.  Then there is a linear stabiliser code
  whose underlying totally isotropic set $S$ is the ideal generated by
  $g(X)h(X,\eta)$ in $\Ra$ and hence is cyclic. Furthermore, the
  centerliser $\centraliser{S}$ of $S$ is the ideal generated by
  $h(X,\eta)$.
\end{theorem}

\begin{proof}
  Let the polynomials be as in the theorem. Clearly the ideal $S$
  generated by $g(X)h(X,\eta)$ is separately cyclic. All that remains
  is to show that $S$ is isotropic. Consider the ring
  $\mathbb{F}_4[X]/X^n -1$. Define via Chinese remaindering an element
  $a(X)$ such that $a(X) = 0\mod g$ and for all irreducible factor
  $r(X,\eta)$
  \[
  a(X) = \left\{
  \begin{array}{ll}
    \eta \mod r(X,\eta) & \textrm{ if } r(X,\eta) \nmid h(X,\eta),\\
    \eta' \mod r(X,\eta) & \textrm{ if } r(X,\eta) \mid h(X,\eta),
  \end{array}\right.
  \]

  We first show that $a(X)$ is a polynomial over $\mathbb{F}_2[X]$
  instead of $\mathbb{F}_4[X]$. To prove this consider the action of
  the Frobenius $\sigma$ on $\Ra$ defined by
  $\sigma\left(u(X,\eta)\right) = u(X,\eta')$. It is sufficient to
  show that $\sigma(a) = a$.  Consider any irreducible factor $r$ of
  $X^n -1$ over $\mathbb{F}_4$.  Since $h(X,\eta)$ contains in it
  exactly one of the factors $r$ or $\sigma(r)$, let us assume,
  without loss of generality, that $r \nmid h(X,\eta)$ and $a$ modulo
  $r$ and $\sigma(r)$ are $\eta$ and $\eta'$ respectively. Therefore
  $\sigma(a)$ modulo $\sigma(r)$ and $\sigma^2(r) = r$ are $\eta'$ and
  $\eta$ modulo respectively. Therefore $a = \sigma(a) \mod
  r\sigma(r)$ for every irreducible factor of $X^n -1/g$ Furthermore
  since $a = 0 \mod g$, we have $\sigma(a) = a$ over $\R(\eta)$.
  
  It is easy to verify that the ideal $S$ as a separately cyclic set
  is generated by the triple $(g,ag,0)$. By
  Theorem~\ref{thm:poly-characterisation}, $S$ is isotropic if and
  only if the equation
  \begin{equation}\label{eqn:gga}
  g(X)g(X^{-1}) a(X) = g(X)g(X^{-1})a(X^-1).
  \end{equation}
  holds modulo $X^n -1$. By Chinese remaindering it is sufficient to
  verify Equation~\ref{eqn:gga} modulo $g$ and $X^n -1/g$ separately.
  
  Clearly Equation~\ref{eqn:gga} holds modulo $g$.  Since $g(X)$ is a
  product of irreducible factors of $X^n -1$ over $\mathbb{F}_2$, for
  any root $\beta$ of $X^n -1$, we have $g(\beta^{-1}) = 0$ if and
  only if $g(\beta) = 0$ (Lemma~\ref{lem:degree}).  By both $g(X)$ and
  $g(X^{-1})$ are invertible modulo $X^n -1/g$. Therefore we need to
  show that $a(X^{-1}) = a(X)$ modulo $X^n -1 /g$. It follows from the
  construction of $a(X)$ that it satisfies the equation $a^2 = a +1
  \mod X^n -1/g$. Therefore $a(X^{-1})\=a(X)^{4^m} = a(X) + 2m$ modulo
  $X^n -1/g$ and hence $a(X^{-1}) = a(X) \mod X^n -1/g$ as $2m=0$ in
  characteristic $2$.

\end{proof}

We call the codes that are characterised by the above theorem as
\emph{$4^m+1$-codes} and the pair of polynomials $(g(X),h(X,\eta)$ the
\emph{generating pair} of the code. We have the following theorem on
the dimension of the code.

\begin{theorem}
  Let $\mathcal{C}$ be $4^m+1$-code generated by the pair $(g,h)$ then
  the stabiliser dimension is given by $\mathrm{deg}(g)$.
\end{theorem}
\begin{proof}
  The set $S$ is an ideal of $\mathbb{F}_4[X]/X^n -1$. Let $d_1$ and
  $d_2$ denote the degree of $g$ and $h$ respectively then clearly $\#
  S$ is $4^{n - (d_1 + d_2)}$. Also note that $g h \sigma(h) = X^n
  -1$. Therefore $d_1 + 2 d_2 = n$. 

  If $k$ is the $\mathbb{F}_2$ dimension of $S$ then $\# S =
  2^k$. Comparing we have $k = 2n - 2d_2 - 2 d_2$. Thus we have $k = n
  - d_1$. Therefore by \ref{thm:isotropic-main-properties} we have the
  required result.
\end{proof}

We now recall some concepts from the theory of classical cyclic codes.

\begin{definition}[BCH distance]
  Let $g(X)$ be a factor of the polynomial $X^n - 1$ over the finite
  field $\mathbb{F}_q$ for some prime power $q$ and let $n$ coprime to
  $q$. The \emph{BCH distance} is the largest integer $d$ such that
  the consecutive distinct powers $\beta^{l}$,$\beta^{l+1},\ldots,
  \beta^{l + d -2}$ are roots of $g$, for some primitive $n$-th root
  $\beta$.
\end{definition}

For a $4^m +1$-code generated by the pair $(g,h)$, by the \emph{BCH
  distance} we mean the BCH distance of $h(X,\eta)$ as a factor of
$X^n -1$ over $\mathbb{F}_4$. We have the following theorem about the
distance.

\begin{theorem}
  The distance of a $4^m+1$ code is at least its BCH distance.
\end{theorem}
\begin{proof}
  The centraliser $\centraliser{S}$ of the underlying totally
  isotropic subspace $S$ is an ideal generated by $h(X,\eta)$ and
  hence can be thought of as classical cyclic code over
  $\mathbb{F}_4$. Hence the $\mathbb{F}_4$-weight of any non-zero
  element in $\centraliser{S}$ is at least the BCH distance. As the
  distance of the code is the the minimum weight of $\centraliser{S}
  \setminus S$ (Theorem~\ref{thm:isotropic-main-properties}), we have
  the desired result.
\end{proof}

To demonstrate the construction for specific cases take $m=1$. The
polynomial $X^5-1 = (X -1) (X^4+X^3+X^2+X+1)$ over
$\mathbb{F}_2$. Further, over $\mathbb{F}_4$, the degree $4$
irreducible factor factorises into $(X^2+\eta X + 1)$ and $(X^2 +
\eta' X + 1)$. If we pick $g = X - 1$ and $h(X,\eta)$ to be any one of
the factors, we get a $[[5,1,3]]$ code which turns out to be the
Laflamme code.

The case $m=2$ is more interesting.  The polynomial $X^{17} -1$
factorises into three factors.
\[
x^{17}-1=(x+1)(x^8 + x^7 + x^6 + x^4 + x^2 + x +1)(x^8+x^5+x^4+x^3+1).\]

We have two possibilities for $g(X)$ here. In one case $g(X) = X -1$
and in the other is $X -1$ times one of the degree $8$ factors. By
choosing $h(X,\eta)$ appropriately in these cases we get a
$[[17,1,7]]$ and a $[[17,9,1]]$ code respectively. We skip the
details for lack of space. To the best of our knowledge, these are new
codes.

\subsection{Decoding $4^m+1$-codes within the BCH limit}

Let $\mathcal{C}$ be $4^m +1$-code with BCH distance $d = 2t +1$. Much
like in the classical case, we show that there is an efficient quantum
algorithm to correct any quantum error of weight at most $t$. Here by
efficient we mean polynomial in the code length $n = 4^m +1$. There
are two key algorithms that we use: (1) The quantum phase finding
algorithm and (2) The Berlekamp decoding algorithm for classical BCH
codes.

The decoding algorithm for classical BCH code can be seen as follows:

\begin{theorem}[Berlekamp]\label{thm:berelekemp-recast}
  Let $g(X)$ be a factor of $X^n - 1$ of BCH distance $d =2 t +1$ over
  a finite field $\mathbb{F}_q$, $q$ and $n$ coprime. Let $e(X)$ be
  any polynomial of weight at most $t$ over $\mathbb{F}_p$. Given a
  polynomial $r(X) = e(X) \mod g(X)$, there is a polynomial time
  algorithm to recover $e(X)$.
\end{theorem}
\begin{proof}[Proof sketch]
  Since $r(X) = e(X) \mod g(X)$ we have $r(X) = e(X) + c(X)$ for some
  $c(X)$ in the ideal generated by $g(X)$. Note that $c(X)$ is a
  valid, but as of now unknown, code word in the code generated by
  $g(X)$. We can think of the computational task of recovering $e(X)$
  as that task of recovering the sent message $c(X)$ for a received
  message $r(X)$ for which there is an efficient algorithm due to
  Berlekamp~\cite[p-98,6.7]{vLin}.
\end{proof}

For the rest of the section fix a $4^m +1$-code $\mathcal{C}$ with BCH
distance $d = 2t +1$. Let $n =4^m +1$ denote its length and assume
that $(g,h)$ is a generating pair for $\mathcal{C}$. Assume that we
have transmitted a quantum message $\ket{\varphi} \in \mathcal{C}$
over the quantum channel and received the corrupted state $\ket{\psi}
= U_\mathbf{a} V_\mathbf{b} \ket{\varphi}$, where the vectors
$\mathbf{a}$ and $\mathbf{b}$ are unknown to us. If we can design an
algorithm that recovers $\mathbf{a}$ and $\mathbf{b}$ without actually
disturbing $\ket{\psi}$, then we have an error correction algorithm as
we recover the sent message by applying the inverse map
$V_\mathbf{b}^\dag U_\mathbf{a}^\dag$ on $\ket{\psi}$. We show that
this is possible provided the joint weight
$\weight{\mathbf{a},\mathbf{b}} \leq t$.

Consider the polynomial $e(X,\eta) = a(X^{-1}) + \eta b(X^{-1})$ as a
polynomial over $\mathbb{F}_4[X]/X^n -1$. Clearly, the
$\mathbb{F}_4$-weight of $e(X,\eta)$ is also at most $t$ and
polynomials $a(X)$, $b(X)$ can be recovered once $e(X,\eta)$ is
recovered. We prove that $e(X,\eta)$ can be recovered modulo
$h(X,\eta)$.

Let $S$ be the underlying totally isotropic set associated with
$\mathcal{C}$. Since $(g,h)$ is the generating pair for $S$, the
factor $g(X)h(X,\eta)$ of $X^n -1$ generates $S$ as an ideal of
$\mathbb{F}_4[X]/X^n -1$. We have the following proposition

\begin{proposition}\label{prop:compute-poly-iso}
  For any $(\mathbf{c},\mathbf{d})$ in $S$, there is an efficient
  quantum algorithm to compute the polynomial $d(X) a(X^{-1} -
  c(X)b(X^{-1})$.
\end{proposition}
\begin{proof}
  Recall that the code $\mathcal{C}$ is the set of vectors stabilised
  by the corresponding Gottesman subgroup $\mathcal{S}$. Let $U =
  \zeta U_{\mathbf{c}} V_{\mathbf{d}}$ be the element in $\mathcal{S}$
  corresponding to the pair $(\mathbf{c},\mathbf{d}) \in S$. It can be
  easily show that $\ket{\psi}$ is an eigen vector of $U$ with eigen
  value $(-1)^{\transpose{\mathbf{d}}\mathbf{a} -
    \transpose{c}\mathbf{b}}$ and using phase finding one can recover
  the inner product $(-1)^{\transpose{\mathbf{d}}\mathbf{a} -
    \transpose{c}\mathbf{b}}$. Repeating the algorithm with
  $(N^k\mathbf{c},N^k\mathbf{d})$, all the inner products
  $\transpose{\mathbf{d}}N^k \mathbf{a} - \transpose{\mathbf{c}}N^k
  \mathbf{b}$ can be recovered. Since these are precisely the
  coefficients the polynomial $d(X) a(X^{-1} - c(X)b(X^{-a})$ modulo
  $X^n -1$, this is sufficient to prove the claim.
\end{proof}

Since $\tilde{g} = gh$ generate the set $S$ as an ideal, both
$\tilde{g}$ and $\eta\tilde{g}$ belong to $S$. Using the algorithm in
Proposition~\ref{prop:compute-poly-iso} with $\tilde{g}$ and
$\eta\tilde{g}$, it is straight forward but tedious to show that the
polynomial $e(X,\eta)$ can be computed modulo $h(X,\eta)$. We are now
in the setting of Theorem~\ref{thm:berelekemp-recast} where
$h(X,\eta)$ as a polynomial in $\mathbb{F}_4[X]/X^n -1$ playing the
role of the generator polynomial. Since $h(X,\eta)$ has BCH distance
$2t+1$, we can recover $e(X,\eta)$ and hence $(\mathbf{a},\mathbf{b})$
using the Berlekamp algorithm. Thus we have the following theorem

\begin{theorem}
  Let $\mathcal{C}$ be a $4^m+1$ code of length $n=4^m +1$ and BCH
  distance $d = 2t+1$. There is quantum algorithm that takes time
  polynomial in $n$ to correct errors of weight at most $t$.
\end{theorem}

\section{Cyclic codes that are not stabiliser}
\label{sec:nonstabiliser}

In this section we give examples for certain nonstabiliser codes that
are cyclic. There has been some work in the construction of
nonstabiliser quantum code. Rains \etal~\cite{rains97nonadditive}
used computer search to construct the first example of a $((5,6,2))$
quantum code which is not stabiliser. Shortly, Roychowdhury and
Vatan~\cite{vwani99nonadditive} gave few more examples of such codes.
Arvind \etal~\cite{arvind:2004:nonstabilizer} gave a generic method
to construct quantum codes for Gottesman subgroups of the error group,
some of which turn out to be nonstabiliser. We summarise their result
in the following proposition.

\begin{proposition}[\cite{arvind:2004:nonstabilizer}]
  \label{prop:pseudo-stab}
  Let $\mathcal{S}$ be a Gottesman subgroup of the error group. Then
  \begin{enumerate}
  \item For any character $\chi$ of $\mathcal{S}$, the set
    $\mathcal{S}_\chi = \{ \chi(s) s | s \in \mathcal{S} \}$ is also a
    Gottesman subgroup of the error group with $P_\chi = \sum_{s\in
      \mathcal{S}} \chi(s) s$ as the corresponding stabiliser code.
  \item The codes $P_\chi$ and $P_\varphi$ are orthogonal unless $\chi
    = \varphi$.
  \item An element $A$ in the algebra $\mathbb{C}[\mathcal{S}]$ is a
    projection if and only if there is a subset $B$ of characters of
    $\mathcal{S}$ such that $A = \sum_{\chi\in B}P_\chi$.
  \end{enumerate}
\end{proposition}

We call the codes thus generated from Gottesman subgroups,
\emph{pseudo-stabiliser} codes.  In particular, Arvind \emph{et
  al}~\cite{arvind:2004:nonstabilizer} have show that the $((5,6,2))$
code of Rains \etal~\cite{rains97nonadditive} is a
pseudo-stabiliser code.

Consider any Gottesman subgroup $\mathcal{S}$ that is separately
cyclic.  Then the corresponding stabiliser code $\mathcal{C}$ is
cyclic. For any character $\chi$ of $\mathcal{S}$, $\mathcal{S}_\chi$
is also separately cyclic as the underlying totally isotropic subspace
$S$ of $\mathbb{F}_p^{2n}$ is same for $\mathcal{S}$ and
$\mathcal{S}_\chi$. Hence the corresponding stabiliser code
$\mathcal{C}_\chi$ is cyclic and therefore $\N^\dag P_\chi \N =
P_\chi$. By Proposition~\ref{prop:pseudo-stab}, any pseudo-stabiliser
code with support in $\mathcal{S}$ is given by the sum of projections
$P_B = \sum_{\chi\in B} P_\chi$. Thus we have $\N^\dag P_B \N =
P_B$. As a result we have the following proposition.

\begin{proposition}
  A pseudo-stabiliser code whose underlying Gottesman subgroup is cyclic is
  also cyclic.
\end{proposition}

The above proposition gives us ways to construct cyclic
pseudo-stabiliser code by first constructing a cyclic stabiliser code.
In particular, the $((5,6,2))$ code is a pseudo-stabiliser code whose
underlying Gottesman subgroup is separately cyclic.  This gives a
concrete example for a nonstabiliser cyclic quantum code.

\bibliography{bibdata} 
\bibliographystyle{abbrv}

\end{document}